\documentclass[11pt,reqno]{amsart}

\usepackage{amssymb,amsmath,amsthm,amscd,latexsym,amsfonts}
\usepackage{mathtools}
\usepackage[T1]{fontenc}
\usepackage{graphicx,epstopdf}
\usepackage{graphicx}
\usepackage{xcolor}
\usepackage{cite}

\usepackage[a4paper,top=3cm, bottom=3cm, left=3.1cm, right=3.1cm]{geometry}

\newtheorem{thm}{Theorem}
\newtheorem{defn}{Definition}
\newtheorem{lemma}{Lemma}
\newtheorem{pro}{Proposition}
\newtheorem{rk}{Remark}
\newtheorem{cor}{Corollary}

\numberwithin{equation}{section} 

\newcommand{\bea}{\begin{eqnarray}}
\newcommand{\eea}{\end{eqnarray}}

\newcommand{\Z}{\mathbb{Z}}

\def\O{\Omega}

\def\O{\Omega}

\def\L{\Lambda}

\def\Z{\mathbb{Z}}

\def \L {\Lambda}

\begin{document}
\title [Gradient Gibbs measures for SOS model]
{Mirror symmetry of height-periodic gradient Gibbs measures of an SOS model on Cayley trees}

\author {U.A. Rozikov}

\address{ U.Rozikov$^{a,b,c}$\begin{itemize}
 \item[$^a$] V.I.Romanovskiy Institute of Mathematics,  9, Universitet str., 100174, Tashkent, Uzbekistan;
\item[$^b$] AKFA University, National Park Street, Barkamol MFY,
Mirzo-Ulugbek district, Tashkent, Uzbekistan;
\item[$^c$] National University of Uzbekistan,  4, Universitet str., 100174, Tashkent, Uzbekistan.
\end{itemize}}
\email{rozikovu@yandex.ru}

\begin{abstract}
 For the  solid-on-solid (SOS) model with spin values from the set of all integers on a Cayley tree we give gradient Gibbs measures (GGMs). Such a measure corresponds to a boundary law (which is an infinite-dimensional vector-valued function defined on vertices of the Cayley tree) satisfying an infinite system of functional equations. We give several concrete GGMs of boundary laws which are independent from vertices of the Cayley tree and (as an infinite-dimensional vector) have periodic, (non-)mirror-symmetric coordinates. Namely, the particular class of height-periodic  boundary laws of period $q\leq 5$ is studied, where solutions are classified by
 their period and (two-)mirror-symmetry. 

\end{abstract}
\maketitle



\section{Introduction}

For a lattice model with a compact set of spin values, it is known that the set of Gibbs measures is non-empty (\cite{BR}, \cite{EHR}, \cite{FV}, \cite{Ge}). But for models with a non-compact set of spin values, the existence problem of the Gibbs measure remains one of the important problems in physics and statistical mechanics, because for systems with an infinite set of spin values the set of Gibbs measures may be empty. In such a situation some gradient Gibbs measures (GGMs) may exist (for detailed motivations and very recent results see \cite{BiKo}, \cite{BEvE}, \cite{HR}, \cite{HKLR}, \cite{HKa}, \cite{HKb}, \cite{Hphd}, \cite{KS}, \cite{Sh}, \cite{Z1}).

This paper is devoted to GGMs of the SOS model with spin values from the set of all integers  on Cayley trees.
 
\subsection{Definitions and known results}
We consider models where spin-configuration $\omega$ is a function from the
vertices of the Cayley tree $\Gamma^k=(V, \vec L)$  to the set $\Z$ of integer numbers, where
$V$ is the set of vertices and $\vec L$ is the set of oriented edges (bonds) of the tree
(see Chapter 1 of \cite{Ro} for properties of the Cayley tree).

For any configuration  $\omega = (\omega(x))_{x \in V} \in \mathbb Z^V$ and edge $b = \langle x,y \rangle$ of $\Gamma^k$
the \textit{difference} along the edge $b$ is given by $\nabla \omega_b = \omega(y) - \omega(x)$, where $\omega_b$ is a configuration on  $b = \langle x,y \rangle$, i.e., $\omega_b=\{\omega(x), \omega(y)\}$.  The configuration $\nabla \omega$ is called the \textit{gradient field} of $\omega$ (see \cite{KS}, \cite{K}).

The gradient spin variables are now defined by $\eta_{\langle x,y \rangle} = \omega(y) - \omega(x)$ for each $\langle x,y \rangle$.

The space of \textit{gradient configurations} is denoted by $\O^\nabla$. The measurable structure on the space $\Omega^{\nabla}$ is given by $\sigma$-algebra $$\mathcal{F}^\nabla:=\sigma(\{ \eta_b \, \vert \, b \in \vec L \}).$$

Let $\mathcal{T}_{\Lambda}^{\nabla}$ be the sigma-algebra of gradient configurations outside of the finite volume $\Lambda$ is generated by all gradient variables outside of $\Lambda$ and the relative height-difference on the boundary of $\Lambda$. 

For nearest-neighboring interaction  potential $\Phi=(\Phi_b)_b$, where
$b=\langle x,y \rangle$ is an edge,  define symmetric transfer matrices $Q_b$ by
\begin{equation}\label{Qd}
Q_b(\omega_b) = e^{- \big(\Phi_b(\omega_b) + | \partial x|^{-1} \Phi_{\{x\}}(\omega(x)) + |\partial y |^{-1} \Phi_{\{y\}} (\omega(y)) \big)},
\end{equation}
where $\partial x$ is the set of all nearest-neighbors of $x$ and $|S|$ denotes the number of elements of the set $S$.

To introduce the notion of \emph{translations} on the Cayley tree
$\Gamma^k$, one uses its group representation $G_k$ which is  the free group with
generators $a_1, \dots,\allowbreak a_{k+1}$ of order $2$ each (i.e.,
$a_i^{-1}=a_i$). It is known (see, for example, \cite[Section~2.2]{Ro})
that the vertices of the Cayley tree are
in a one-to-one correspondence with the elements of the group $G_k$.
Consider the family of \emph{left} shifts $T_g\colon
G_k\to G_k$ ($g\in G_k$) defined by $T_g(a):=g a$, $ a\in G_k$.

This group is used to define translation-invariance of functions defined on the vertices
of the Cayley tree. In particular, the following definition is used
\begin{defn}\label{tid}  The potential $\Phi=(\Phi_b)_b$ is called invariant under a subgroup $\hat G_k\subset G_k$ of
translations if for any $g\in \hat G_k$ and $b=\langle x, y\rangle$ one has
$\Phi_{gb}(g \omega_{b})=\Phi_b(\omega_b)$, where $gb=\langle gx, gy\rangle$ and $g\omega$  is defined by $g\omega(x) =
\omega(g^{-1}x)$, $x\in G_k$.

In the case $\hat G_k= G_k$ the potential is called translation-invariant.
\end{defn}
Similarly one can define translation-invariant Gibbs measures (see \cite[Section 2.3.1]{BR}).

Define the Markov (Gibbsian) specification as
$$
\gamma_\Lambda^\Phi(\sigma_\Lambda = \omega_\Lambda | \omega) = (Z_\Lambda^\Phi)(\omega)^{-1} \prod_{b \cap \Lambda \neq \emptyset} Q_b(\omega_b).
$$

If for any bond $b=\langle x,y \rangle$ the transfer operator $Q_b(\omega_b)$ is
a function of gradient spin variable $\zeta_b=\omega(y)-\omega(x)$ then the underlying potential $\Phi$ is called
a \textit{gradient interaction potential}.

 \emph{Boundary laws} (see \cite{Z1}) which allow to describe the subset of $\mathcal{G}(\gamma)$ of all Gibbs measures.

\begin{defn}\label{def} 
	
	\begin{itemize}
\item	A family of vectors $\{ l_{xy} \}_{\langle x,y \rangle \in \vec L}$ with $l_{xy}=\left(l_{xy}(i) : i\in \Z\right) \in (0, \infty)^\Z$ is called a {\em boundary law for the transfer operators $\{ Q_b\}_{b \in \vec L}$} if for each $\langle x,y \rangle \in \vec L$ there exists a constant  $c_{xy}>0$ such that the consistency equation
	\begin{equation}\label{eq:bl}
	l_{xy}(i) = c_{xy} \prod_{z \in \partial x \setminus \{y \}} \sum_{j \in \Z} Q_{zx}(i,j) l_{zx}(j)
	\end{equation}
	holds for every $i \in \Z$.
\item  A boundary law $l$ is said to be {\em normalisable} if and only if
\begin{equation}\label{Norm}
	\sum_{i \in \Z} \Big( \prod_{z \in \partial x} \sum_{j \in \Z} Q_{zx}(i,j) l_{zx}(j) \Big) < \infty
\end{equation} at any $x \in V$.

\item 	A boundary law 	is called {\em $q$-height-periodic} (or $q$-periodic) if $l_{xy} (i + q) = l_{xy}(i)$
	for every oriented edge $\langle x,y \rangle \in \vec L$ and each $i \in \Z$.
	\end{itemize}
\end{defn}

It is known that there is an one-to-one correspondence between boundary laws
and tree-indexed Markov chains if the boundary laws are  normalisable \cite{Z1}. In recent papers \cite{GRH}, \cite{HKb} some normalisable boundary laws are found.

In \cite{HKLR},  \cite{HKa}, \cite{HKb}, \cite{KS}
some non-normalisable boundary laws are used to give GGM. Following \cite{KS} and \cite{K} let us give definition of GGM:

\begin{defn}  The gradient Gibbs specification is defined as the family of probability kernels $\left(\gamma_{\Lambda}^{\prime}\right)_{\Lambda \Subset V}$ from $\left(\Omega^{\nabla}, \mathcal{T}_{\Lambda}^{\nabla}\right)$ to $\left(\Omega^{\nabla}, \mathcal{F}^{\nabla}\right)$ such that
$$
\int F(\rho) \gamma_{\Lambda}^{\prime}(d \rho \mid \zeta)=\int F(\nabla \varphi) \gamma_{\Lambda}(d \varphi \mid \omega)
$$
for all bounded $\mathcal{F}^{\nabla}$-measurable functions $F$, where $\omega \in \Omega$ is any height-configuration with $\nabla \omega=\zeta$.
\end{defn}
\begin{defn} 
A probability measure $\nu$ on $\Omega^{\nabla}$ is called a GGM if it satisfies the $DLR$ equation
$$
\int \nu(d \zeta) F(\zeta)=\int \nu(d \zeta) \int \gamma_{\Lambda}^{\prime}(d \tilde{\zeta} \mid \zeta) F(\tilde{\zeta})
$$
for every finite $\Lambda \subset V$ and for all bounded functions $F$ on $\Omega^{\nabla}$. 
\end{defn}

For $\Lambda\subset V$, fix  a site $w \in \Lambda$.
If the boundary law $l$ is assumed to be $q$-height-periodic, then take  $s \in \mathbb{Z}_q=\{0,1,\dots,q-1\}$ and define probability measure $\nu_{w,s}$ on $\mathbb{Z}^{\{b \in \vec L \mid b \subset \Lambda\}}$ by
$$
\nu_{w,s}(\eta_{\Lambda \cup \partial \Lambda}=\zeta_{\Lambda \cup \partial \Lambda})
$$
$$=Z^\Lambda_{w,s}\prod_{y \in \partial \Lambda} l_{yy_\L}\Bigl (T_q(
s+\sum_{b\in \Gamma(w,y)}\zeta_b)
\Bigr) \prod_{b \cap \Lambda \neq \emptyset}Q_b(\zeta_b),
$$
where $Z^\Lambda_{w,s}$ is a normalization constant, $\Gamma(w,y)$ is the unique path from $w$ to $y$
and $T_q: \mathbb{Z} \rightarrow \mathbb{Z}_q$ denotes the coset projection.

\begin{thm} \cite{KS}
	Let $(l_{\langle xy\rangle })_{\langle x,y\rangle  \in \vec L}$ be any $q$-height-periodic boundary law for some gradient interaction potential.
Fix any site $w \in V$ and any class label $s \in \mathbb{Z}_q$. Then
$$	\nu_{w,s}(\eta_{\Lambda \cup \partial \Lambda}=\zeta_{\L\cup\partial\L})
	$$
\begin{equation}
=Z^\Lambda_{w,s} \prod_{y \in \partial \Lambda} l_{yy_\L}\Bigl (T_q(
	s+\sum_{b\in \Gamma(w,y)}\zeta_b)
	\Bigr) \prod_{b \cap \Lambda \neq \emptyset}
	Q_b(\zeta_b)
	\end{equation}
gives a consistent family of probability measures on the gradient space $\Omega^\nabla$.
Here $\Lambda$ with $w \in  \L \subset V$ is any finite connected set,
$\zeta_{\L\cup\partial\L} \in \Z^{\{b \in \vec L \mid b \subset (\L \cup \partial\L)\}}$ and $Z^\Lambda_{w,s}$ is a normalization constant.
\end{thm}
	The measure $\nu_{w,s}$ is called a pinned gradient measure.

If the $q$-height-periodic boundary law  and the underlying potential are translation invariant (see Definition \ref{tid})  then it is possible to obtain a
probability measure $\nu$ on the gradient space by mixing the pinned gradient measures:

\begin{thm}\cite{KS}\label{KF}	
Let a $q$-height-periodic boundary law $l$  and  its gradient interaction potential are translation invariant.
Let $\Lambda \subset V$ be any finite connected set and let $w\in \Lambda$ be any vertex. Then the measure $\nu $ with marginals given by
\begin{equation}
\nu (\eta_{\L\cup\partial \L} = \zeta_{\L\cup\partial\L}) = Z_\L \ \left(\sum_{s\in\Z_q}  \prod_{y \in \partial\L} l \big(s + \sum_{b \in \Gamma(w,y)} \zeta_{b}\big)  \right)\prod_{b \cap \L \neq \emptyset} Q(\zeta_b),
\end{equation}
where $Z_\L$ is a normalisation constant, defines a translation invariant GGM on $\Omega^\nabla$.
\end{thm}
\begin{rk}\label{rk}
Note that (see \cite{HKLR} and \cite{HKb})	if a height-periodic boundary law is obtained from another one by a cyclic
	shift, then it leads to the same GGM.  Moreover, the identifiability result Theorem 5 in \cite{HKb}
	guarantees that distinct (up to cyclic shift and multiplication	by positive constants) boundary laws leads to distinct GGMs.
\end{rk}
In this paper we consider the nearest-neighboring interaction  potential $\Phi=(\Phi_b)_b$, which corresponds to the SOS model and will study GGMs of this model. By Theorem \ref{KF} each $q$-height-periodic boundary law $l$ defines a translation invariant GGM. Some examples of such boundary laws already found (see \cite{HR} - \cite{KS}). In this paper our aim is to find \textit{new} $q$-height-periodic boundary laws for the SOS model and use Theorem \ref{KF} to construct GGMs corresponding to the new boundary laws. 

\subsection{The boundary law equation for the SOS model.}

In this subsection for  $\sigma:x\in V\mapsto \sigma(x)\in \mathbb Z$, consider the Hamiltonian
of the SOS model, i.e.,
\begin{equation}\label{nu1}
H(\sigma)=-J\sum_{\langle x,y\rangle: \atop x,y\in V}
|\sigma(x)-\sigma(y)|,\end{equation}
where $J>0$.

Then by (\ref{eq:bl}) the equation for translation-invariant boundary laws has the following form
\begin{equation}\label{di1}
	z_i=\left({\theta^{|i|}+
		\sum_{j\in \mathbb Z_0}\theta^{|i-j|}z_j
		\over
		1+\sum_{j\in \mathbb Z_0}\theta^{|j|}z_j}\right)^k, \quad i\in\mathbb Z_0:=\mathbb Z\setminus \{0\}.
\end{equation}
Here $\theta=\exp(-J\beta)<1$.

Let $\mathbf z(\theta)=(z_i=z_i(\theta)>0, i\in \mathbb Z_0)$ be a solution to (\ref{di1}).   Denote
\begin{equation}\label{lr}
	l_i\equiv l_i(\theta)=\sum_{j=-\infty}^{-1}\theta^{|i-j|}z_j, \quad
	r_i\equiv r_i(\theta)=\sum_{j=1}^{\infty}\theta^{|i-j|}z_j, \quad i\in\mathbb Z_0.
\end{equation}
It is clear that each $l_i$ and $r_i$ can be a finite positive number or $+\infty$.

\begin{lemma}\label{l1} \cite{HKLR} For each $i\in \mathbb Z_0$ we have
	\begin{itemize}
		\item $l_i<+\infty$ if and only if $l_0<+\infty$;
		
		\item $r_i<+\infty$ if and only if $r_0<+\infty$.
	\end{itemize}
\end{lemma}

\begin{pro}\cite{HKLR} \label{pps: 1} 
	A vector $\mathbf z=(z_i,i\in \mathbb Z)$, with $z_0=1$,  is a solution to (\ref{di1})
	if and only if for $u_i=\sqrt[k]{z_i}$ the following holds
	\begin{equation}\label{Va}
		u_i^k=\frac{{u_{i-1}+u_{i+1}-\tau u_i}}{u_{-1}+u_{1}-\tau}, \quad i\in \mathbb Z,
	\end{equation}
	where $\tau=\theta^{-1}+\theta=2\cosh(\beta)$.
\end{pro}

In \cite{HKLR} and \cite{HR} some 4-periodic boundary laws are found.

\section{$q$-periodic solutions for $q\leq 5$}
The equation (\ref{Va}) is equivalent to the following 
\begin{equation}\label{Re}
	u_{i+1}=(u_{-1}+u_1-\tau)u_{i}^k+\tau u_{i}-
	u_{i-1}, \end{equation}
where $i\geq 0$, $u_{-1}>0$, $u_{0}>0$, $u_{1}>0$ are some initial numbers which (for $\theta<1$, i.e. $J>0$) satisfy 
\begin{equation}\label{ut}
u_0=1, \ \ u_{-1}+u_1<\tau.
\end{equation}

Note that for $i=0$ the above equation is trivially fulfilled for all values of $u_1$ and $u_{-1}$.

 {\bf The main problem} is to find $u_1$ and $u_{-1}$ such that the sequence  $\{u_i\}_{i=-1}^\infty$ generated by (\ref{Re}) is strictly positive. 
\subsection{Bounded sequence}
The aim of this subsection is to prove the following

\begin{pro} If a strictly positive sequence $\{u_i\}_{i=-1}^\infty$ 
	 satisfies (\ref{Re}), (\ref{ut}) then it is bounded as
	$$0<u_i< x_0\min\left\{{k-1\over k^{k/(k-1)}}\tau, 1\right\}, \ \ i\geq 1$$
	where 
	 $x_0=\sqrt[k-1]{\tau\over \tau-u_{-1}-u_1}$.  
\end{pro}
\begin{proof} Assume $u_i$ is not bounded, i.e., for any $M>0$ there exists $i$ such that $u_i>M$.
	Consider the function $\psi(x)=(u_{-1}+u_1-\tau)x^k+\tau x$.  

Since $u_{-1}+u_1<\tau$ the function $\psi(x)$ is monotone increasing in $(0, x_*)$, (where $x_*=x_0\sqrt[k-1]{1/k}$) and decreasing in $ (x_*, +\infty)$. Note that $x_0>x_*$, $\psi(x_0)=0$  and $\psi(x)<0$ for all $x>x_0$. Therefore, for $u_i>M>x_0$ by (\ref{Re}) and conditions of proposition we have  
$$	0<u_{i+1}+u_{i-1}=(u_{-1}+u_1-\tau)u_{i}^k+\tau u_{i}<
(u_{-1}+u_1-\tau)M^k+\tau M<0.$$
This is a contradiction. Hence $u_i$ must be bounded.
The lower bound is 0 (by the assumption of the proposition) and the upper bound is obtained as
$$0<u_{i-1}<u_{i+1}+u_{i-1}=\psi(u_i)\leq \max_{x>0} \psi(x)=\psi(x_*)=x_0{k-1\over k^{k/(k-1)}}\tau$$ and $\psi(u_i)>0$ iff $0<u_i<x_0$.
\end{proof}
	
\subsection{Periodic solutions}
\begin{defn} 
	
\begin{itemize}
	\item A sequence $\{u_i\}_{i=-1}^\infty$ is called  $q$-periodic if 
	$$u_{i+q}=u_i \ \ \mbox{for some} \ \ q\geq 1 \ \ \mbox{and for all} \ \ i\geq -1;$$ 
	\item A $q$-periodic sequence is called  mirror symmetric if 
	$$u_{i}=u_{q-i}, \ \ \mbox{for all} \ \ i=-1, 0, 1, \dots, [q/2].$$  
	\item A $q$-periodic sequence is called two-mirror symmetric if there exists $p$ such that $p<q$ and 
	$$u_i=u_{p-i}, \ \ i=-1, 0, 1, \dots, [p/2];$$
	$$u_{p+j}=u_{q-j}, \ \ j=0, 1, \dots, [(q-p)/2].$$
\end{itemize}
\end{defn}
\begin{pro}\label{pu} Let $k\geq 2$ and $q\geq 1$. Any $q$-periodic sequence $\{u_i\}_{i=-1}^\infty$ with $u_{-1}=u_1$,  $u_0=1$,  satisfying (\ref{Re}), is mirror symmetric too.
\end{pro}
\begin{proof} Any $q$-periodic sequence $\{u_i\}_{i=-1}^\infty$ with $u_{-1}=u_1$ and $u_0=1$ is in the form
$$u_{1}, 1, u_1, u_2, \dots, u_{q-1}=u_{1}, u_{q}=1, u_{q+1}=u_1, u_{q+2}=u_2, \dots$$
Using mathematical induction we prove that
\begin{equation}\label{uq}
	u_{i}=u_{q-i}, \ \ \mbox{for all} \ \ i= 0, 1, \dots, [q/2].
\end{equation}

By (\ref{Re}) we have 
\begin{equation}\label{ms}
	\begin{array}{llllll}
u_2=(2u_1-\tau)u_{1}^k+\tau u_{1}-1\\[2mm]
u_3=(2u_1-\tau)u_{2}^k+\tau u_{2}-u_1\\[2mm]
\vdots \\[2mm]
u_{q-2}=(2u_1-\tau)u_{q-3}^k+\tau u_{q-3}-u_{q-4}\\[2mm]
u_{1}=u_{q-1}=(2u_1-\tau)u_{q-2}^k+\tau u_{q-2}-u_{q-3}\\[2mm]
1=u_{q}=(2u_1-\tau)u_{1}^k+\tau u_{1}-u_{q-2}.\\[2mm]
\end{array}
\end{equation}
We know that for $i=0$ the equation (\ref{Re}) is satisfied for all values of $u_1$ and $u_{-1}$. By assumption $u_{-1}=u_1$ and $q$-periodicity we have $u_1=u_{q-1}$. Moreover, from the first and last equations of the system (\ref{ms}) we get $u_{q-2}=u_2$. Thus the equality (\ref{uq}) is true for $i=0,1,2$. Assume now that it is true for any $i< [q/2]$. We show it for $i+1$:
by (\ref{ms}) we have 
\begin{equation}\label{mss}
	\begin{array}{ll}
		u_{i+1}=(2u_1-\tau)u_{i}^k+\tau u_{i}-u_{i-1},\\[2mm]
u_{q-i+1}=(2u_1-\tau)u_{q-i}^k+\tau u_{q-i}-u_{q-i-1}.
\end{array}
\end{equation}
By assumption of the induction we have $u_{i}=u_{q-i}$ and $u_{i-1}=u_{q-i+1}$, Therefore (\ref{mss}) gives $u_{i+1}=u_{q-i-1}.$ This completes the proof.
\end{proof}
\begin{pro}\label{po} Let $k\geq 2$ and $q\geq 1$. If $\{u_i\}_{i=-1}^\infty$ is a $q$-periodic sequence  with $u_{-1}\ne u_1$, $u_0=1$,  satisfying (\ref{Re}), and $u_p=1$, $u_{p+1}=u_{-1}$ for some $p<q$ then it is two-mirror symmetric.
\end{pro}
\begin{proof} Under conditions of the proposition the $q$-periodic sequence $\{u_i\}_{i=-1}^\infty$ is in the form
	$$u_{-1}, 1, u_1, u_2, \dots, u_{p-1}, u_{p}=1, u_{p+1}=u_{-1}, u_{p+2}, \dots u_{q-1}=u_{-1}, u_q=1, u_{q+1}=u_1, \dots$$
Note that the recurrence equation (\ref{Re}) 	is given in "forward" form. This equation can be written as a "backward" form: 
	\begin{equation}\label{Reb}
		u_{i-1}=(u_{-1}+u_1-\tau)u_{i}^k+\tau u_{i}-
		u_{i+1}. \end{equation}
	For $i=p$ from (\ref{Reb}), by conditions of the proposition, we get
	$$u_{p-1}=(u_{-1}+u_1-\tau)u_{p}^k+\tau u_{p}-
	u_{p+1}=(u_{-1}+u_1-\tau)+\tau -u_{-1}=u_1.$$
	Now using mathematical induction one can prove that 
\begin{equation}\label{pp}	u_i=u_{p-i}, \ \ i=0, 1, \dots, [p/2].
	\end{equation}
By conditions $u_{p}=1$, $u_{p+1}=u_{-1}$ and equation (\ref{Re}) for $i=p+2$ we get 
$$u_{p+2}=(u_{-1}+u_1-\tau)u_{p+1}^k+\tau u_{p+1}-
u_{p}=(u_{-1}+u_1-\tau)u_{-1}^k+\tau u_{-1} -1$$
and using equation (\ref{Reb}) and $q$-periodicity of the sequence we obtain
$$u_{q-2}=(u_{-1}+u_1-\tau)u_{q-1}^k+\tau u_{q-1}-
u_{q}=(u_{-1}+u_1-\tau)u_{-1}^k+\tau u_{-1} -1.$$
Hence, $u_{p+2}=u_{q-2}$. Again using the mathematical induction one shows 
	\begin{equation}\label{ppp}u_{p+j}=u_{q-j}, \ \ j=0, 1, \dots, [(q-p)/2].
	\end{equation}

\end{proof}
Keeping in mind the equalities (\ref{uq}), (\ref{pp}), (\ref{ppp}) from Propositions \ref{pu} and \ref{po} we obtain the following

\begin{cor}
If a $q$-periodic sequence is a solution to (\ref{Re}) and it is a mirror (or two-mirror) symmetric then the sequence consists of up to $[q/2]$ distinct elements (including $u_0=1$).   
\end{cor}

For $q=1,2,3,4,5$ we separately study $q$-periodic solutions of (\ref{Re}).

\subsubsection{Case $u_{-1}=u_1$.}
In this subsection we consider initial values $u_{-1}=u_1$, $u_0=1$.
\
{\it Sub-case:} $q=1$. In this case the only solution is 
$$u_{i}=1, \ \ i=-1,0,1,\dots$$

{\it  Sub-case:} $q=2$ in this case the sequence is
$$u_{1}, 1,  u_1, 1, u_1, 1, \dots$$
Full description of such sequences is given in Proposition 4.5 of \cite{HKLR}. 

{\it  Sub-case:} $q=3$. The 3-periodic, mirror symmetric sequence has the form
$$u_{1}, 1,  u_1, u_1, 1, u_1, u_1, 1, \dots$$
where $x:=u_1$ satisfies
\begin{equation}\label{x3}
	2x^{k+1}-\tau x^k+(\tau-1)x-1=0.
\end{equation}
Following \cite{HKLR} (see pages 20-21), we note that this equation has the solution $x=1$ independently of the parameters $\tau$, $k$. 
Denote
$$\tau_0:={2k+1\over k-1}.$$

The following lemma gives the full analysis of the equation (\ref{x3}):

\begin{lemma}\label{l6} \cite{HKLR} For each $k\geq 2$,
	there is exactly one critical value of $\tau=2\cosh(\beta)$, called $\tau_c=\tau_c(k)$, such that
	\begin{itemize}
		\item[1.] $\tau_c<\tau_0$;
		\item[2.]  if $\tau<\tau_c$ then (\ref{x3}) has unique  positive solution;
		\item[3.] if $\tau=\tau_c$ then the equation has two positive solutions;
		\item[4.] if $\tau>\tau_c$, $\tau\ne \tau_0$
		then it has exactly three solutions;
		\item[5.] if $\tau=\tau_0$, then the equation has two solutions.
		
		In each case one of solutions  is $x=1$.
	\end{itemize}
\end{lemma}

 {\it  Sub-case:} $q=4$. The 4-periodic, mirror symmetric sequence has the form
 $$u_{1}, 1,  u_1, u_2, u_1, 1, u_1, u_2, u_1, 1, \dots$$
 where $x:=u_1$, $y:=u_2$ satisfy the following system
\begin{equation}\label{xyu} \begin{array}{ll}
x=(2x-\tau)y^k+\tau y-x,\\[2mm]
y=(2x-\tau)x^k+\tau x-1.
\end{array}
\end{equation}
Note that $y=1$ is a solution of the first equation of (\ref{xyu}), independently on values of $x$. But this solution gives a 2-periodic sequence mentioned above. Therefore, we assume $y\ne 1$ and define 
\begin{equation}\label{fu}
	f(y)={\tau y(1+y+\dots+y^{k-2})\over 2(1+y+\dots+y^{k-1})}, \ \ g(x)=(2x-\tau)x^k+\tau x-1.
	\end{equation}
Then the system (\ref{xyu}) can be written as 
\begin{equation}\label{fg} 
	x=f(g(x)).
	\end{equation}
It seems difficult to solve this equation for $k\geq 3$. Here we give all solutions in the case $k=2$. By a computer analysis one can find explicitly the following positive solutions of (\ref{fg}) for $k=2$:
$$\begin{array}{ll}
	x_1=x_1(\tau)={1\over 4}\left(\tau+\sqrt{\tau^2-4\tau+ 4\sqrt{\tau^2-4\tau}}\right), \\[2mm] 
x_2=x_2(\tau)={1\over 4}\left(\tau-\sqrt{\tau^2-4\tau+ 4\sqrt{\tau^2-4\tau}}\right),
\end{array} \, \mbox{for} \, \tau\geq 4;\ \
 \, \mbox{and}$$
\begin{equation}\label{chor} 
 \begin{array}{ll}
x_3=x_3(\tau)={1\over 4}\left(\tau+\sqrt{\tau^2-4\tau- 4\sqrt{\tau^2-4\tau}}\right), \\[2mm] 
x_4=x_4(\tau)={1\over 4}\left(\tau-\sqrt{\tau^2-4\tau -4\sqrt{\tau^2-4\tau}}\right),
\end{array} \, \mbox{for} \, \tau\geq 2(1+\sqrt{5}).
\end{equation}
It is easy to check that $2x_i-\tau<0$, $i=1,2,3,4$.  
Moreover, the corresponding $y_i=g(x_i)>0,$ $i=1,2,3,4$ (see Fig. \ref{gxi}). 
\begin{figure}[h!]
	\includegraphics[width=9cm]{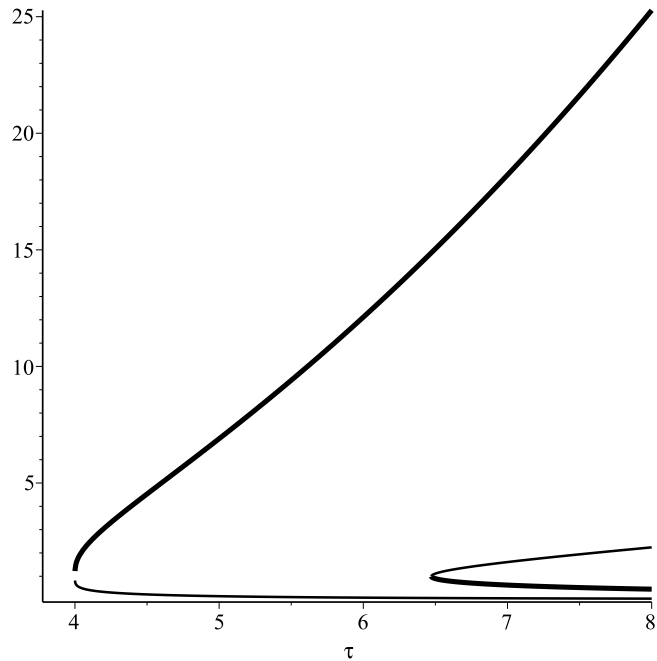}
	\caption{ The graphs of the function $g(x_1(\tau))$ (bold curve on $[4, +\infty]$); $g(x_2(\tau))$ (thin curve on $[4, +\infty]$, $g(x_3(\tau))$ (thin curve on $[2(1+\sqrt{5}), +\infty]$); $g(x_4(\tau))$ (bold curve on $[2(1+\sqrt{5}), +\infty]$ ).}\label{gxi}
\end{figure}

 {\it  Sub-case:} $q=5$. The 5-periodic, mirror symmetric sequence has the form
$$u_{1}, 1,  u_1, u_2, u_2, u_1, 1, u_1, u_2, u_2, u_1, 1, \dots$$
where $x:=u_1$, $y:=u_2$ satisfy the following system
\begin{equation}\label{xy5} \begin{array}{ll}
		x=(2x-\tau)y^k+\tau y-y,\\[2mm]
		y=(2x-\tau)x^k+\tau x-1.
		\end{array}
\end{equation}
Note that $x=y=1$ is a solution of this system, which is known in cases $q=1,2,3,4$ too. Introduce 
$$\phi(x)=(2x-\tau)(g(x))^k+(\tau -1)g(x), $$
where $g(x)$ is defined in (\ref{fu}). Thus we have to find positive fixed points of $\phi(x)$. In general, it is difficult to solve the equation $x=\phi(x)$, which is polynomial equation of degree $k^2+k+1$. Even for the case $k=2$ the equation has degree 7. One solution is $x=1$. Following Fig. \ref{tau8} and Fig. \ref{g8}, for $\tau=8$, one can see that the equation may have maximum (i.e. seven) number of positive solutions. 

Thus for $k=2$, $\tau=8$ there are exactly seven 5-periodic sequences (one of them is 1-periodic, corresponding to $x=y=1$.)
 \begin{figure}[h!]
 	\includegraphics[width=8cm]{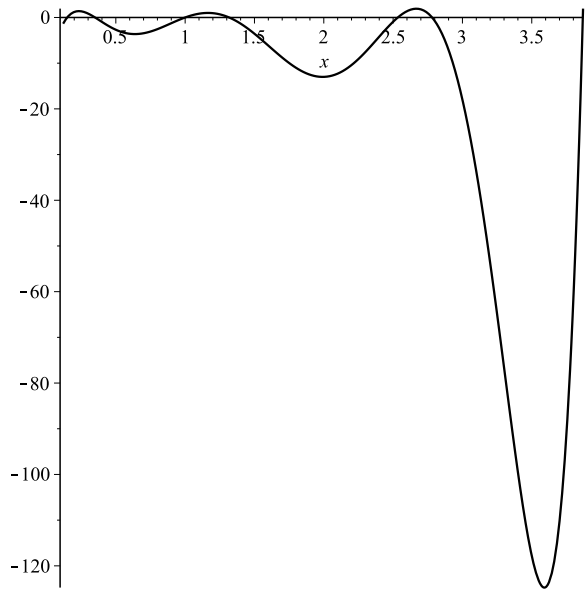}
 	\caption{ The graph of the function $\phi(x)-x$ for the case $\tau=8$, and $x\in (0.12, 3.87)$. Shown all seven roots. }\label{tau8}
 \end{figure}
\begin{figure}[h!]
	\includegraphics[width=8cm]{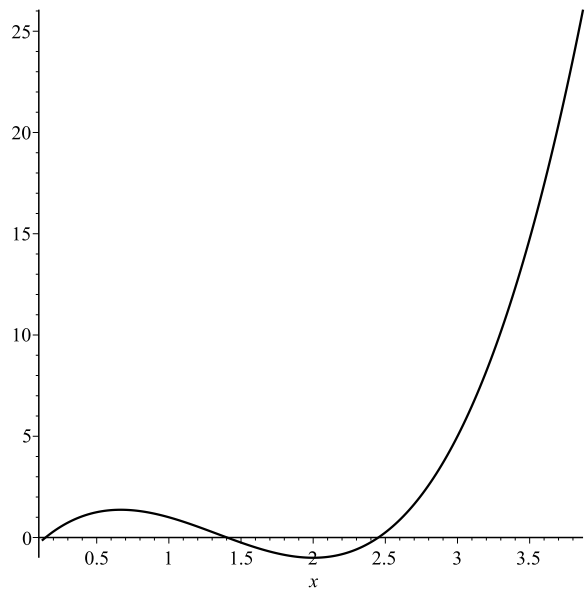}
	\caption{ The graph of the function $g(x)$ for the case $\tau=8$, and $x\in (0.12, 3.87)$. Showing that at each  root mentioned in Fig.\ref{tau8}, the corresponding value of $y=g(x)$ is positive.}\label{g8}
\end{figure}

\subsubsection{Case $u_{-1}\ne u_1$.}
In this case a $q$-periodic sequence, satisfying (\ref{Re}), may exist only for $q\geq 3$.

{\it Sub-case:}  $q=3$. In this sub-case a $3$-periodic, but non-mirror symmetric sequence has the form:
$$u_{-1}, 1, u_1, u_{-1}, 1, u_1, u_{-1}, 1, u_1,\dots$$
where $x:=u_{-1}$, $y:=u_1$ satisfy the following system
\begin{equation}\label{xyy} \begin{array}{ll}
		x=(x+y-\tau)y^k+\tau y-1,\\[2mm]
		y=(x+y-\tau)x^k+\tau x-1.
	\end{array}
\end{equation}
In this system $x$ and $y$ are symmetrical and by our assumption 
 $x\ne y$. In system (\ref{xyy}), $y=1$ (resp. $x=1$) is a solution for the first (resp. second) equation independently, on value of $x$ (resp. $y$). For $y=1$ from the second equation of the system we get 
 \begin{equation}\label{x3a}
 	x^{k+1}-(\tau-1) x^k+\tau x-2=0.
 \end{equation}
 Denoting $t=1/x$ from the equation (\ref{x3a}) we get the equation (\ref{x3}) with respect to $t$. Thus Lemma \ref{l6} is true for solutions $t$, giving $x=1/t$. Consequently, 
 the equation (\ref{x3a}) has up to three solutions one of which is 1. Similarly, taking $x=1$ from the first equation of (\ref{x3}) we get up to two (except 1) solutions.
 Thus if $x=1$ or $y=1$ we have that the system (\ref{xyy}) has up to 5 solutions in the form:
\begin{equation}\label{5x}
(1,1), (1,x_1) , (1,x_2), (x_1,1), (x_2,1).\end{equation}

For the case $k=2$ we have explicit values of $x_1$ and $x_2$: 
 $$x_{1,2}= {1\over 2}\left(\tau-2\pm \sqrt{\tau^2-4\tau-4}\right), \ \ \tau\geq 2(1+\sqrt{2}).$$

 Assume now that $x\ne 1$ and $y\ne 1$ with $x\ne y$. Then for $k=2$, from the first equation of the system (\ref{xyy}) we get 
$$ x=(x+y-\tau)y^2+\tau y-1
 \Rightarrow  x(1-y^2)=y^3-1+\tau y(1-y)
 \Rightarrow x(1+y)=-(1+y+y^2)+\tau y.
 $$
 Similarly from the second equation we get 
 $y(1+x)=-(1+x+x^2)+\tau x$. Subtracting from the first of these equations the second one, we obtain 
 $x+y=\tau$.   
 
 By similar computations, for $k=3$ and $k=4$,  one gets the following equations:
 $$(x+y)^2-\tau(1+x+y)+x+y=0, \ \ \mbox{for} \ \ k=3.$$ 
 $$(x+y-\tau)(x^2+xy+y^2+x+y+1)=0, \ \ \mbox{for} \ \ k=4.$$ 
Each of these equations has unique positive solution $x+y=\tau$, which is impossible because of our assumption $x+y<\tau$. Thus in case $k=2$, $k=3$ and $k=4$ there are only solutions (\ref{5x}).
 
{\it Sub-case:}  $q=4$. One possibility is 
\begin{equation}\label{up}
	u_{-1}, 1, u_1, 1, u_{-1}, 1, u_1, \dots.
	\end{equation}
In \cite{HKLR}, for $k=2$, this kind of 4-periodic sequences are described. In \cite{HR}, the case $k\geq 3$ is studied, in particular, the following theorem is proved:
\begin{thm} For the SOS model  on the Cayley tree of order $k=3$
	there are critical values $\tau^{(1)}_{\mathrm{cr}}\approx 3.13039$,  $\tau^{(2)}_{\mathrm{cr}}\approx 4,01009$ such that the following assertions hold
	\begin{enumerate}
		\item If $\tau \leq \tau_{\mathrm{cr}}^{(1)}$ then there is precisely one GGM associated to a boundary law of the type \eqref{up}.
		\item If $\tau\in (\tau_{\mathrm{cr}}^{(1)},4]$ then there are precisely two such GGMs.
		\item If $\tau\in (4, \tau_{\mathrm{cr}}^{(2)}]\cup\{3\sqrt{2}\}$  then there are at most three such GGMs.
		\item If $\tau\in (\tau_{\mathrm{cr}}^{(2)}, +\infty)\setminus\{3\sqrt{2}\}$  then there are at
		most four such measures associated to boundary laws of the type \eqref{up}.
	\end{enumerate}
\end{thm}

Another kind of 4-periodic sequence is 
$$u_{-1}, 1, u_1, u_2, u_{-1}, 1, u_1, u_2, u_{-1}, 1, u_1, u_2, \dots.$$
Here $x:=u_{-1}$, $y:=u_1$, $z:=u_2\ne 1$ satisfy the following
\begin{equation}\label{z4} \begin{array}{lll}
		x=(x+y-\tau)z^k+\tau z-y\\[2mm]
		1=(x+y-\tau)x^k+\tau x -z\\[2mm]
		z=(x+y-\tau)y^k+\tau y-1.
		\end{array}
\end{equation}

Take $x=1$ then from the second equation we get $z=y$. Then from the first (or third) equation one gets 
\begin{equation}\label{zy4} 
		y=(1+y-\tau)y^k+\tau y-1 \ \ \Leftrightarrow \ \ y^{k+1}-(\tau-1)y^k+(\tau-1)y-1=0.
\end{equation}		
This equation has the solution $y=1$ independently of the parameters $\tau, k$. Dividing both sides by $y-1$
we get
\begin{equation}\label{uy22}
y^k-(\tau-2)\sum_{j=1}^{k-1}y^j+1=0.
\end{equation}
The equation (\ref{uy22}) has again the solution $y=1$ iff $\tau=\tau_1$, where
$$\tau_1:={2k\over k-1}.$$

It is well known (see \cite{Pra}, p.28) that the number of positive
roots of the polynomial (\ref{uy22}) does not exceed the number of sign
changes of its coefficients.
Since  $\tau>2$, the number of positive roots of
the polynomial (\ref{uy22}) is at most 2.

The following lemma gives the full analysis of the equation (\ref{uy22}):

\begin{lemma}\label{y6} For each $k\geq 2$ the following assertions hold
	\begin{itemize}
		\item[1.]  if $\tau<\tau_1$ then (\ref{uy22}) has no positive solution;
		\item[2.] if $\tau=\tau_1$ then the equation has a unique positive solution, denoted by $y_1$;
		\item[4.] if $\tau>\tau_1$,
		then it has exactly two solutions, denoted $y_1$, $y_2$;
		\end{itemize}
\end{lemma}
\begin{proof} It is very similar to the proof of Lemma 4.7 in \cite{HKLR} and Lemma 1 in \cite{HR}.
\end{proof}	
By this lemma and the symmetry of $x$ and $y$ we have the following solutions of (\ref{z4}):
\begin{equation}\label{xs} 
	(1,1,1), (1,y_1,y_1), (1,y_2,y_2),
(y_1,1,y_1), (y_2, 1, y_2).
\end{equation}
It seems difficult to solve system (\ref{z4}) for $k\geq 3$ and $x\ne 1$, $y\ne 1$. But for $k=2$ it is possible to give all solutions explicitly. Indeed, subtracting the third equation from the second one  we get
$$(x+y-\tau)(x^2-y^2)+\tau (x-y)=0 \Leftrightarrow (x-y)[(x+y-\tau)(x+y)+\tau]=0.$$
Since $x\ne y$ the last equality gives
$$x+y={1\over 2}(\tau\pm\sqrt{\tau(\tau-4)}), \ \ \tau\geq \tau_1=4.$$
Using these values of $x+y$ from the first equation of (\ref{z4}) we get (since $z\ne 1$),
$$ z={1\over 2}(\tau-2+\sqrt{\tau(\tau-4)}), \ \ \mbox{for} \ \ x+y={1\over 2}(\tau+\sqrt{\tau(\tau-4)}).$$
Consequently, 
$$x=z={1\over 2}(\tau-2+\sqrt{\tau(\tau-4)}), \ \ y=1.$$
and 
$$z={1\over 2}(\tau-2-\sqrt{\tau(\tau-4)}), \ \ \mbox{for} \ \ x+y={1\over 2}(\tau-\sqrt{\tau(\tau-4)}),$$
$$x=1, \ \ y=z={1\over 2}(\tau-2-\sqrt{\tau(\tau-4)}).$$

{\it Sub-case:}  $q=5$.
In general, a 5-periodic sequence  is 
$$u_{-1}, 1, u_1, u_2, u_3, u_{-1}, 1, u_1, u_2, u_3, u_{-1}, \dots,$$
where $x:=u_{-1}$, $y:=u_1$, $z=u_2$, $t=u_3$ satisfy the following
\begin{equation}\label{x5g} \begin{array}{llll}
		z=(x+y-\tau)y^k+\tau y -1\\[2mm]
		t=(x+y-\tau)z^k+\tau z-y\\[2mm]
		x=(x+y-\tau)t^k+\tau t-z\\[2mm]
		1=(x+y-\tau)x^k+\tau x-t.
	\end{array}
\end{equation}
In this system of equations we consider the following particular cases:

{\it Sub-sub-case:} we assume $z=y$ and $t=1$ then the system is reduced to
\begin{equation}\label{x5} \begin{array}{ll}
		1=(x+y-\tau)x^k+\tau x -1\\[2mm]
		y=(x+y-\tau)y^k+\tau y-1.
	\end{array}
\end{equation}
Note that in this system of equations $x=1$ if and only if $y=1$. Assume now that $x\ne 1$ and $y\ne 1$.
Define 
$$\xi(x)=x^{-k}(1-(\tau-1)x+\tau x^k-x^{k+1}), $$
$$\eta(x)=x^{-k}(2-\tau x+\tau x^k-x^{k+1}). $$

For $x$ satisfying $\eta(x)>0$, the system (\ref{x5}) is reduced to find positive solutions of $x=\xi(\eta(x))$. In general, it is difficult to solve this equation, which is a polynomial equation of degree $k^2+k+1$. For the case $k=2$ the equation has degree 7, one solution is $x=1$. Dividing by $x-1$ we get the following equation
$$
\zeta(x):=	(2\tau-1)x^6-(3\tau^2-3\tau+2)x^5+(\tau^3+6\tau-2)x^4$$
	\begin{equation}\label{66}
	-(\tau^3+2\tau^2+4\tau)x^3+(6\tau^2-4\tau+8)x^2-(12\tau-8)x+8=0.
	\end{equation}
Since $\tau>2$ the number of sign changes of coefficients in LHS of (\ref{66}) is 6. 
Therefore it may have up to six positive solutions. 

 Following Fig. \ref{t51}, for $\tau=8$, one can see that the equation may have maximum (i.e. seven) number of positive solutions, including 1. Moreover, at each these solutions the corresponding value of $\eta(x)$ is positive. Thus for $k=2$, $\tau=8$ there are exactly seven 5-periodic sequences, each has the form 
 $$x, 1, y, y, 1, x, 1, y, y, 1, x, \dots$$
 
\begin{figure}[h!]
	\includegraphics[width=7cm]{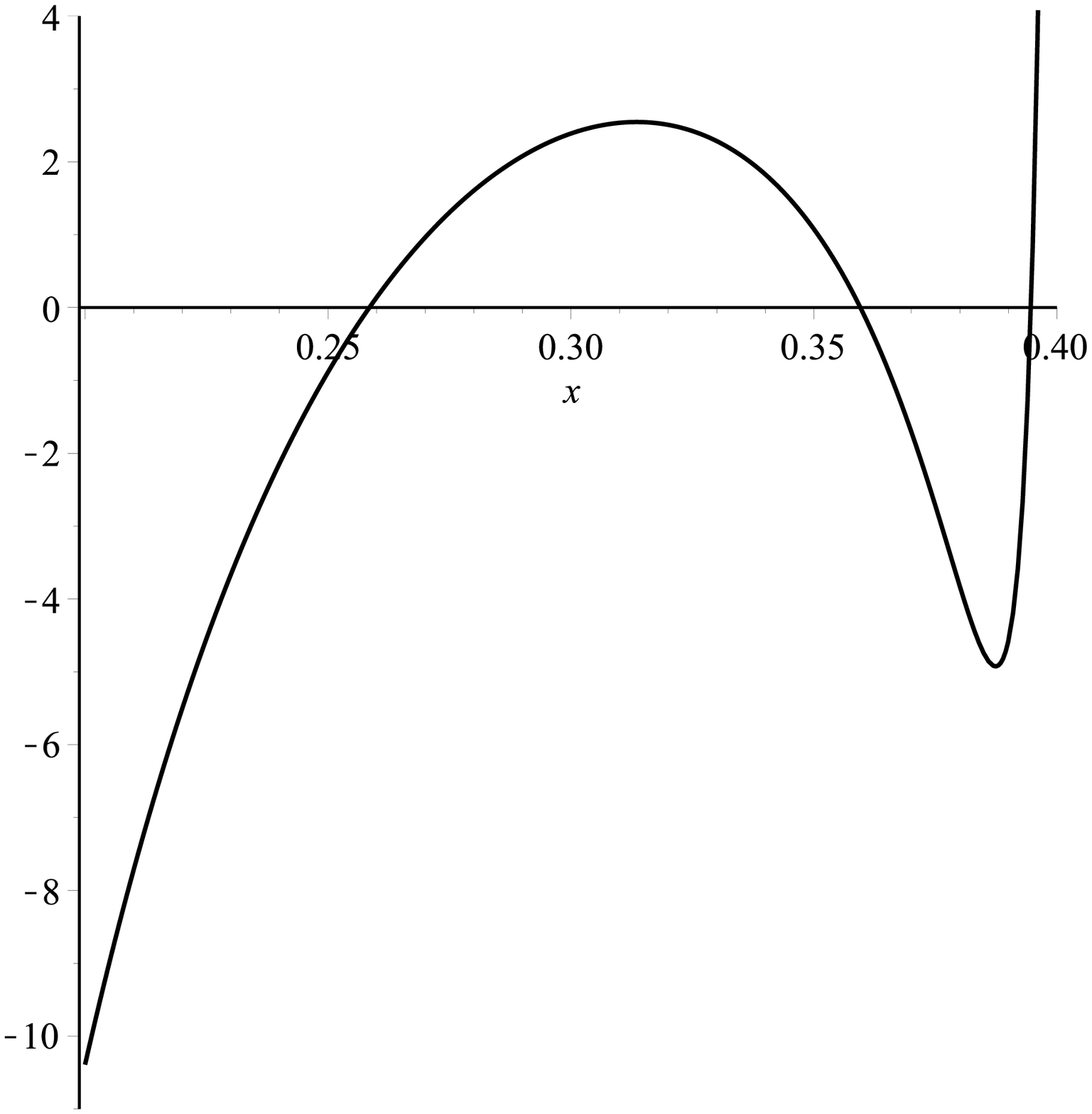}
%	\caption{ The graph of the function $phi(x)-x$ for the case $\tau=8$, and $x\in (0.12, 3.87)$. Shown all seven roots. }\label{t51}
%\end{figure}
%\begin{figure}[h!]
	\includegraphics[width=7cm]{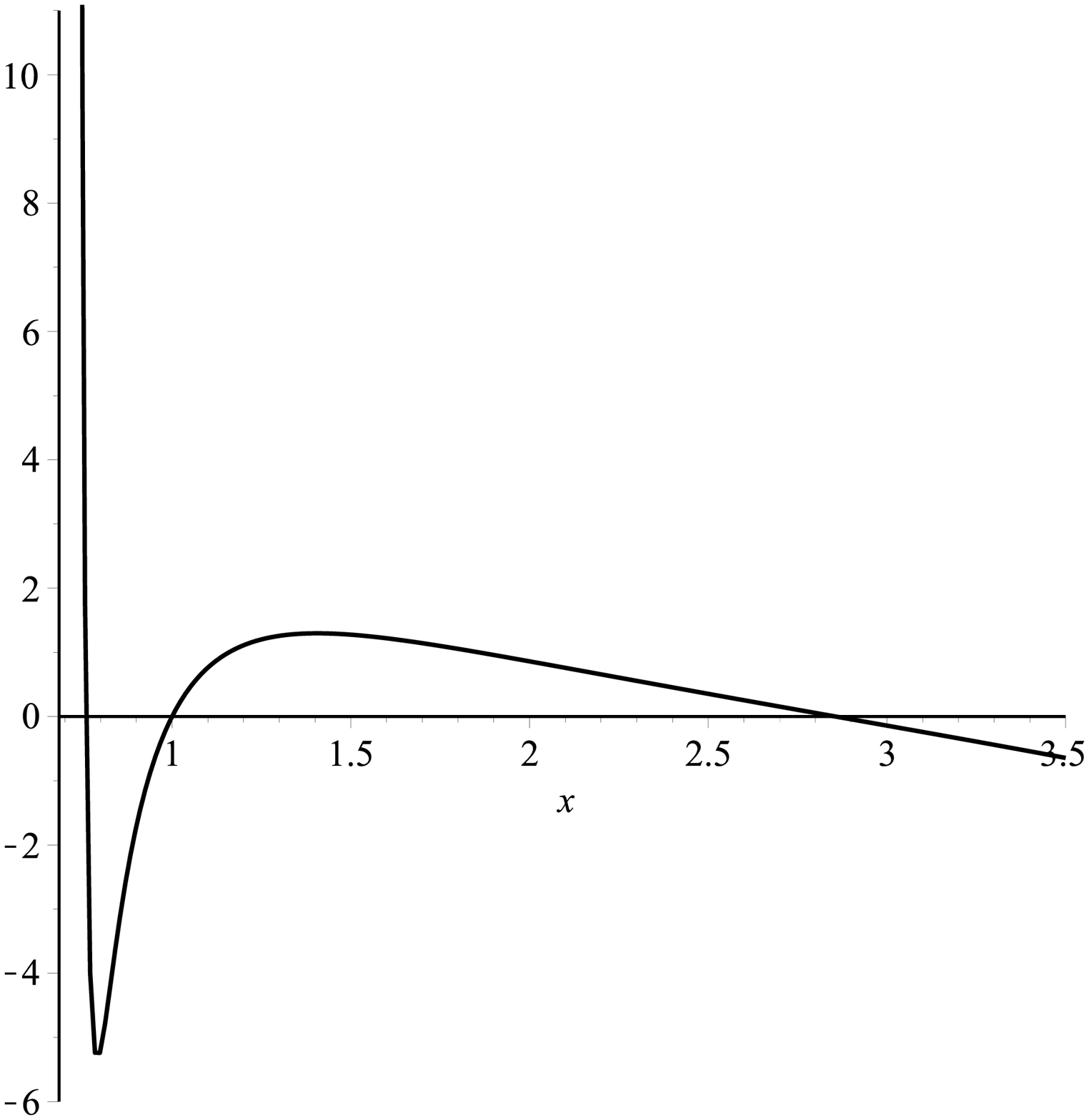}
	\includegraphics[width=7cm]{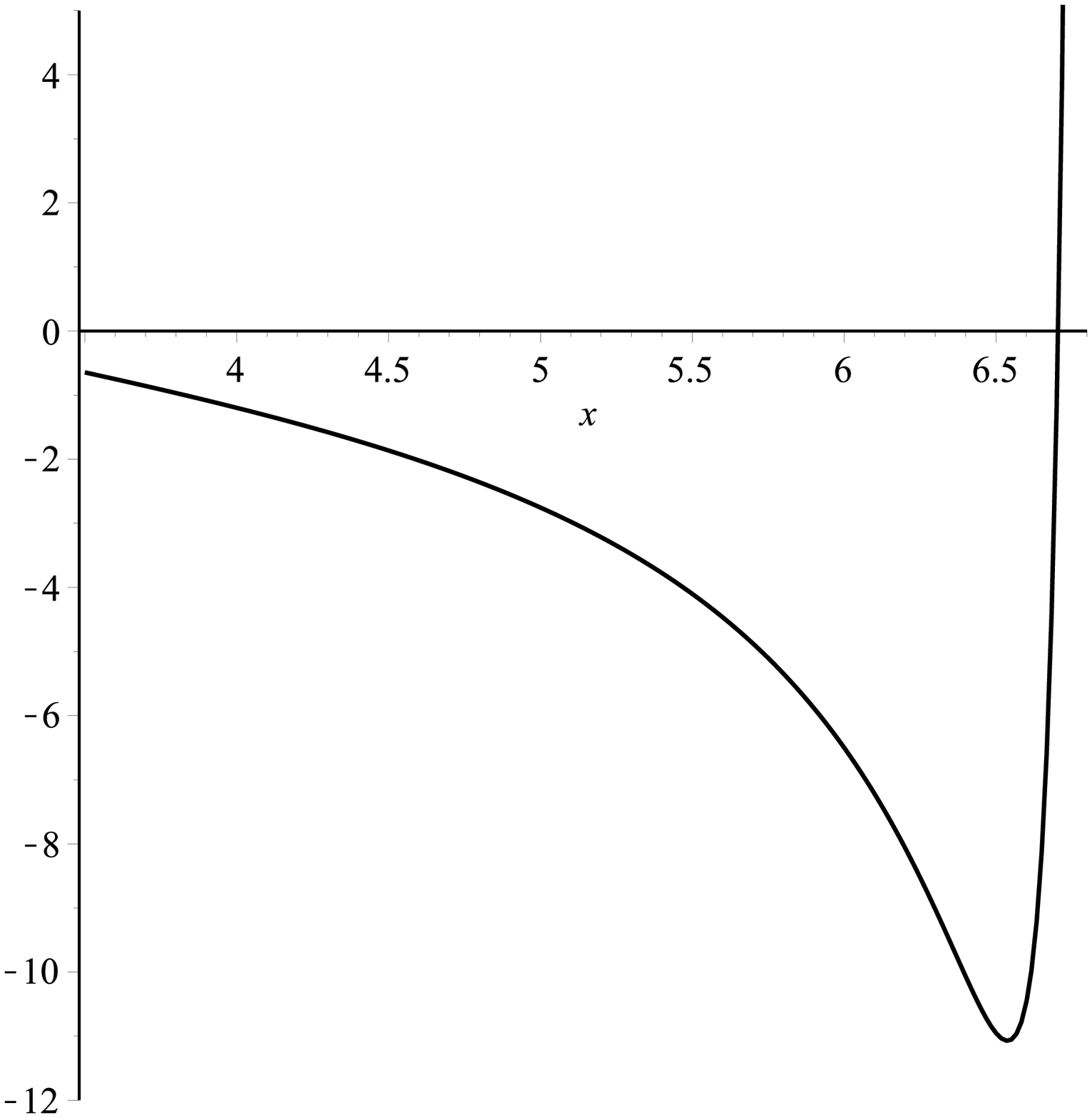}
	\caption{ The graph of the function $\zeta(x)$ for the case $\tau=8$, and $x\in (0.2, 0.4)\cup (0.7, 3.5)\cup (3.5, 6.8)$ separately to show all 7 solutions.}\label{t51}
\end{figure}

{\it Sub-sub-case:} we assume $z=1$ and $t=x$ then the system is reduced to
\begin{equation}\label{x5y} \begin{array}{ll}
		x=(x+y-\tau)x^k+\tau x -1\\[2mm]
		1=(x+y-\tau)y^k+\tau y-1.
	\end{array}
\end{equation}
We note that by exchanging $x$ and $y$ in (\ref{x5y})
 we obtain (\ref{x5}). Therefore these two systems are identical (up to rename of unknowns.) Such 5-periodic sequences have the form
 $$x, 1, y, 1, x, x, 1, y, 1, x, x, \dots$$
which coincide with sequences  of the previous sub-sub-case.

Recall that by Theorem \ref{KF} to each $q$-height-periodic boundary law corresponds a GGM and using the identifiability results mentioned in Remark \ref{rk}  we obtain the following main result of our paper:

\begin{thm} For the SOS model (\ref{nu1})  on the Cayley tree of order $k\geq 2$ the following assertions hold
		\begin{enumerate}
		\item[1.] If $k=2$ then\footnote{All GGMs mentioned in this theorem are different from such measures of \cite{HKLR} and \cite{HR}.}
		\begin{itemize}
		\item[1.1.] for $\tau\in (4, 2(1+\sqrt{5}))$ there are exactly two 4-height-periodic mirror symmetric GGMs ;
		\item[1.2.] if $\tau=2(1+\sqrt{5})$ then there are exactly three 4-height-periodic mirror symmetric GGMs.
		\item[1.3] if $\tau>2(1+\sqrt{5})$  then there are exactly four such GGMs. {\rm (These measures correspond to the solutions (\ref{chor}).)}
		\end{itemize}
		\item[2.] If $k=2$, $\tau=8$  then there are exactly six 5-height-periodic mirror symmetric GGMs, which correspond to the solutions (except 1) shown in Fig. \ref{tau8}.
		\item[3.] For each $k\geq 2$, there is critical value $\tau_c=\tau_c(k)$ such that
		\begin{itemize}
			\item[3.1.] if $\tau=\tau_c$ or $\tau=\tau_0={2k+1\over k-1}$ then there is one  3-height-periodic, non-mirror symmetric GGMs;
			\item[3.2.] if $\tau>\tau_c$, $\tau\ne \tau_0$
			then there are two 3-height-periodic, non-mirror symmetric GGMs. {\rm (These measures correspond to the solutions (\ref{5x}).)}
		\end{itemize}
		
		In particular, if $k=2$ then $\tau_c(2)=2(1+\sqrt{2})$.
		\item[4.] For each $k\geq 2$ there is critical value 
		$t_1={2k\over k-1}$ such that 
		\begin{itemize}
			\item[4.1.] if $\tau=\tau_1$ then there is one  
			4-height-periodic non-mirror symmetric GGMs;
			\item[4.2.] if $\tau>\tau_1$ then there are two  
			such GGMs. {\rm (These measures correspond to the solutions (\ref{xs}).)}
			\end{itemize} 
		 \item[5.] If $k=2$, $\tau=8$  then there are six 5-height-periodic non-mirror symmetric  GGMs.
	\end{enumerate}
\end{thm}

\section*{Data availability statements}
 The datasets generated during and/or analysed during the current study 
 are available from the author (U.A.Rozikov) on reasonable request.

\section*{ Acknowledgements}
The author thanks both referees for their useful and helpful comments. 

The work supported by the fundamental project (number: F-FA-2021-425)  of The Ministry of Innovative Development of the Republic of Uzbekistan.

\end{document}